\documentclass[a4paper,11pt,final]{article}

\usepackage[margin=1.1in]{geometry}

\usepackage[utf8]{inputenc}

\usepackage{amsmath}

\usepackage{amsthm}

\usepackage{amssymb}

\usepackage{bm}

\usepackage{dsfont}

\usepackage{graphicx}

\begin{document}

\newtheorem{lema}{Lemma}
\newtheorem{teo}{Theorem}


\title{{\bf Lecture notes on Legendre polynomials: their origin and main properties}}

\author{F\'{a}bio M.~S.~Lima \\
{\small Institute of Physics, University of Bras\'{i}lia, Campus Universitário `Darcy Ribeiro',} \\
{\small Asa Norte, 70919-970, Bras\'{i}lia-DF, Brazil} \\
{\small E-mail: fmsl@unb.br}}



\date{\today}

\maketitle

\section*{Abstract}
It is well-known that separation of variables in 2nd order partial differential equations (PDEs) for physical problems with spherical symmetry usually leads to Cauchy's differential equation for the radial coordinate and Legendre's differential equation for the polar angle $\theta$. For eigenvalues of the form $\,n\,(n+1)$, $n \ge 0\,$ being an integer, Legendre's equation admits certain polynomials $P_n(\cos{\theta})$ as solutions, which form a complete set of continuous orthogonal functions for all $\theta \in [0,\pi]$. This allows us to take the polynomials $P_n(x)$, where $x = \cos{\theta}$, as a basis for the Fourier-Legendre series expansion of any function $f(x)$ continuous by parts over $\,x \in [-1,1]$.  These lecture notes correspond to the end of my course on Mathematical Methods for Physics, when I did derive the differential equations and solutions for physical problems with spherical symmetry. For those interested in Number Theory, I have included an application of shifted Legendre polynomials in \emph{irrationality proofs}, following a method introduced by Beukers to show that $\zeta{(2)}$ and $\zeta{(3)}$ are irrational numbers.

\vspace{0.5cm}

\noindent Keywords: Separation of variables in PDEs, Spherical symmetry, Legendre polynomials, Beukers' integrals, Irrationality proofs.


\section{Introduction}
Let $f: D \rightarrow \mathds{R}$ be a function of three variables twice differentiable over a domain $\,D \subseteq \mathds{R}^3$. As is well-known, the \emph{Laplacian operator} is defined as $\,\nabla^2{f} := \bm{\nabla} \cdot \left(\bm{\nabla}{f} \right)$, where $\,\bm{\nabla}\,f = \frac{\partial f}{\partial x} \, \mathbf{i} +\frac{\partial f}{\partial y} \, \mathbf{j} +\frac{\partial f}{\partial z} \, \mathbf{k}\,$ is the gradient of $f$ and $\,\bm{\nabla} \cdot \vec{F} = \frac{\partial F_x}{\partial x} +\frac{\partial F_y}{\partial y} +\frac{\partial F_z}{\partial z}\,$ is the divergence of a vector field $\,\vec{F} = F_x \: \mathbf{i} +F_y \: \mathbf{j} +F_z \: \mathbf{k}$, all in Cartesian coordinates $(x,y,z)$. Therefore,
\begin{eqnarray}
\nabla^2{f} = \left( \frac{\partial }{\partial x} \, , \, \frac{\partial }{\partial y} \, , \, \frac{\partial }{\partial z} \right) \cdot \left( \frac{\partial f}{\partial x} \, , \, \frac{\partial f}{\partial y} \, , \, \frac{\partial f}{\partial z} \right) \nonumber \\
= \frac{\partial^2 f}{\partial x^2} +\frac{\partial^2 f}{\partial y^2} +\frac{\partial^2 f}{\partial z^2} \, .
\end{eqnarray}
In \emph{spherical coordinates} $(r,\theta,\phi)$, $r$ is the distance to the origin $O$, chosen at the center of the spherical `source', $\theta$ is the \emph{polar} angle, and $\phi$ is the \emph{azimuthal} angle, as indicated in Fig.~\ref{fig1}, below. These coordinates are related to the Cartesian ones by $\,x = r \, \sin{\theta} \, \cos{\phi}$, $y = r \, \sin{\theta} \, \sin{\phi}$, and $\,z = r \, \cos{\theta}$.\footnote{Inversely, $\,r = \sqrt{x^2+y^2+z^2}$, $\theta = \arccos{\left(z/\sqrt{x^2+y^2+z^2}\right)}$, and $\,\phi = \arctan{(y/x)}$.}  For a one-to-one correspondence between these coordinates systems, with a \emph{unique} representation for each point, we must take $\,r \in [0, \infty)$, $\theta \in [0, \pi]$, and $\,\phi \in [0, 2 \pi)$, all angles in radians. In spherical coordinates, the Laplacian operator reads
\begin{equation}
\nabla^2{f} = \frac{1}{\,r^2} \, \frac{\partial}{\partial r} \left( r^2 \, \frac{\partial f}{\partial r} \right) +\frac{1}{\,r^2 \, \sin{\theta}} \, \frac{\partial }{\partial \theta} \left( \sin{\theta} \, \frac{\partial f}{\partial \theta} \right) +\frac{1}{\,r^2 \, \sin^2{\theta}} \, \frac{\,\partial^2 f}{\,\partial \phi^2} \: .
\label{eq:lapla}
\end{equation}
For a proof of this expression using matrix rotations, rather than the tedious change of variables usually found in textbooks, see Ref.~\cite{Lima2022}. Equation~\eqref{eq:lapla} is useful in many physical problems, as, for instance, when we want to solve Poisson's equation (1813) $\,\bm{\nabla} \cdot \vec{E} = -\,\nabla^2 V = \rho(r)/\epsilon_0\,$ in order to determine the electrostatic potential $V$ of a charged sphere, or the heat equation (Fourier, 1822) $\,\nabla^2 T = \frac{c}{k} \, \frac{\partial T}{\partial t}\,$ to find the temperature $T$ around a sphere, or the Poisson equation $\,\bm{\nabla} \cdot \vec{g} = -\nabla^2 V = -\,4 \pi \, G\, \rho\,$ for the gravitational potential $V$ around a star, or the Schr\"{o}dinger equation (1926) $\,-\,\hbar^2/(2 m)\,\nabla^2 \psi +V(r)\,\psi = E \, \psi\,$ in order to find the eigenfunctions $\psi_n(r,\theta,\phi)$ and the corresponding energy eigenvalues $E_n$ for the electron in a Hydrogen atom, etc.

\begin{figure}[hbt]
\centering
\scalebox{0.25}{\includegraphics{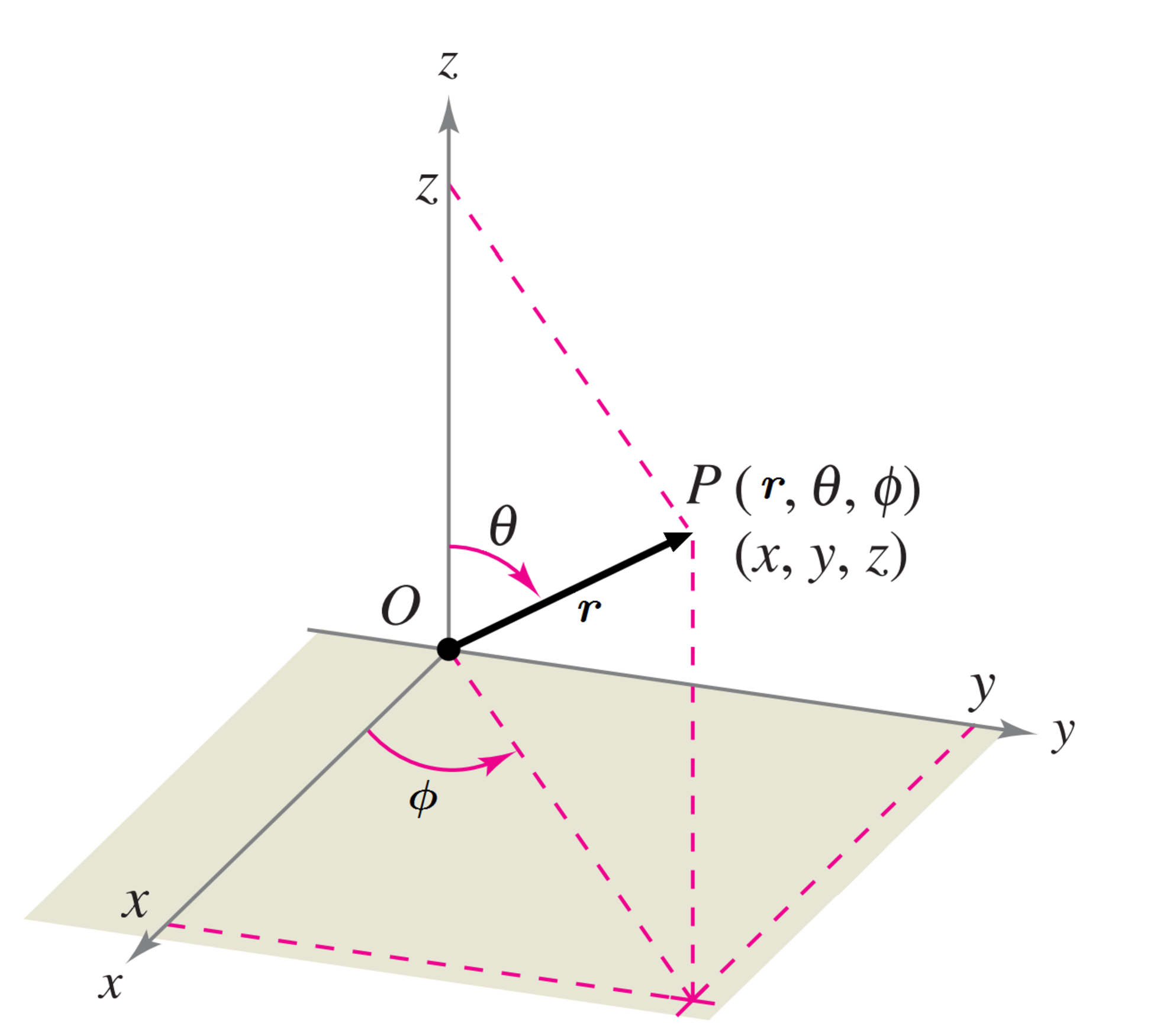}} \label{fig1}
\caption{\small{Cartesian and spherical coordinates. As usual, $r$ is the distance of (arbitrary) point $P=(x,y,z)$ to the origin $O=(0,0,0)$, $\theta$ is the polar angle (between $\mathbf{k}$ and the vector $\,\vec{r} := x\,\mathbf{i} +y\,\mathbf{j} +z\,\mathbf{k}$), and $\phi$ is the azimuthal angle (between $\mathbf{i}$ and the projection of vector $\vec{r}$ onto the $x y$-plane).}}
\end{figure}

We are interested here in Boundary-Value Problems (BVPs) with \emph{spherical symmetry} in which nothing changes when $\phi$ (only) changes, which is just where Legendre's functions and polynomials arise. In this case, the Laplacian operator in Eq.~\eqref{eq:lapla} reduces to
\begin{equation}
\nabla^2{f} = \frac{1}{r} \, \frac{\partial^2 }{\partial r^2} \left( r \, f\right) +\frac{1}{\,r^2 \, \sin{\theta}} \, \frac{\partial }{\partial \theta} \left( \sin{\theta} \, \frac{\partial f}{\partial \theta} \right) \, ,
\label{eq:PDErtheta}
\end{equation}
where we have made use of the fact that $\,\frac{1}{r^2} \, \frac{\partial}{\partial r} \left( r^2 \, \frac{\partial f}{\partial r} \right) = \frac{1}{r} \, \frac{\partial^2 }{\partial r^2} \left( r \, f\right)$. Let us take as our prototype a charged sphere of radius $a$ (made of a dielectric material), whose charge distribution varies only with $r$ (from $0$ to $a$) and $\theta$ (from $0$ to $\pi$). For $r>a$, the charge density is null, so the electrostatic potential $V(r,\theta)$ for external points will be the solution of the \emph{homogeneous} Poisson's equation, i.e.~Laplace's partial differential equation (PDE) $\,\nabla^2 V = 0$, established in 1789. Assuming that the potential is known for all points at the surface $r=a$, being given by $\,V(a,\theta) = F(\theta)$, which is our boundary condition, and that it remains \emph{finite} everywhere, one has a well-poised linear BVP, for which the method of \emph{separation of variables} is known to work. From Eq.~\eqref{eq:PDErtheta}, Laplace's equation applied to this problem becomes
\begin{equation}
\frac{1}{r} \, \frac{\partial^2 }{\partial r^2} \left( r \, V \right) +\frac{1}{\,r^2 \, \sin{\theta}} \, \frac{\partial }{\partial \theta} \left( \sin{\theta} \, \frac{\partial V}{\partial \theta} \right) = 0 \, .
\end{equation}
On assuming that the solution $\,V(r,\theta)\,$ has the form $\,R(r) \cdot W(\theta)$, one finds, after a multiplication by $\,r^2/(R\,W)$,
\begin{equation}
r\,\frac{\,(r\,R)''}{R} = -\,\frac{\,\left(\sin{\theta} \; W' \right)'}{\sin{\theta} \; W} \, ,
\end{equation}
where the left-hand side depends only on $r$ and the right-hand side depends only on $\theta$. Since $r$ and $\theta$ are independent variables, the above equality is possible \emph{only if both sides are equal to a constant}, which for convenience we choose as $\,\lambda \ge 0$. This leads to two ordinary differential equations (ODEs), one for the \emph{radial} part, namely
\begin{equation}
r\,\left(r\,R\right)'' -\,\lambda\,R = 0 \, ,
\label{eq:radial}
\end{equation}
and the other for the \emph{angular} part, namely
\begin{equation}
\left(\,\sin{\theta} \; W' \,\right)' +\lambda \: \sin{\theta} \;\, W = 0 \, .
\label{eq:LegTrig}
\end{equation}

The radial ODE is Cauchy's equation, which is easily solved by substituting $\,r = e^{\,\rho}$. This exponential substitution changes Eq.~\eqref{eq:radial} to $\,S'' +2 S' -\lambda\,S = 0$, where $S(\rho) = R(r)$. This is a $2$nd-order ODE with constant coefficients, whose solution is simply
\begin{equation}
S(\rho) = A\,e^{(-1 +k)\,\rho} +B\,e^{(-1 -k)\,\rho} \, ,
\end{equation}
where $\,k = \sqrt{1 +\lambda}$. This solution corresponds to $\,R(r) = A \: r^{-1 +k} +B\,r^{-1 -k}$. Since $\,k>1$, the first term increases without limit when $\,r \rightarrow \infty$, contrarily to the fact that the solution $V(r,\theta)$ should remain finite for all $r>0$. In order to guarantee this finiteness, we choose $\,A=0$, so the \emph{radial solution} reduces to
\begin{equation}
R(r) = \frac{B}{\:r^{\,1+\sqrt{1 +\lambda\,}}\,} \; , \; \forall \, r>a \, .
\end{equation}

For a better identification of Eq.~\eqref{eq:LegTrig}, let us make the trigonometric substitution $\,x = \cos{\theta}$. Since $\,\frac{du}{d\theta} = \frac{dx}{d\theta} \; \frac{du}{dx} = -\,\sin{\theta} ~ \frac{du}{dx} = -\,\sqrt{1 -x^2\,} ~ \frac{du}{dx}$, the angular ODE simplifies to
\begin{equation}
\left[\left(1 -x^2 \right) \, y' \, \right]' +\lambda \; y = 0 \, ,
\label{eq:origLegendre}
\end{equation}
where it is understood that $\,W(\theta) = y(x)\,$ and, since $\theta \in [0,\pi]$, then $\,x \in [-1,1]$. This is just \emph{\textbf{Legendre's differential equation}}. In textbooks, it often appears as
\begin{equation}
\left(1 -x^{2\,} \right) y'' -2 x \: y' +\lambda \; y = 0 \, ,
\label{eq:Legendre}
\end{equation}
an equivalent form in which $\,\lambda\,$ sometimes is substituted by $\,\ell \, (\ell +1)$, $\ell \ge 0$, which is a valid procedure since, so far, $\lambda$ is just a non-negative \emph{real} number (without further restrictions). Up to Eq.~\eqref{eq:Legendre}, there are no homogeneous boundary conditions available from which we could determine a discrete set of eigenvalues $\lambda$, as occurs e.g.~in the standing waves in a string of length $L$ tied at the endpoints, for which the eigenfrequencies are $\,f_n = n \, f_1$, where $\,f_1 = v/(2 L)\,$ and $\,v = \sqrt{F/\mu}\,$ is the speed of propagation of the wave (Taylor's law, 1721), but this does not mean that it should be a \emph{continuous} variable, as occurs in Fourier transforms. Our physical problem is finite and we should not treat $\lambda$ as an infinite limit of the Fourier series (which would lead to Fourier integrals). In the next section, we'll show that the application of Frobenius method to Eq.~\eqref{eq:Legendre} generates a power series solution known as \emph{Legendre's functions} and, among these functions, there are finite solutions in the form of polynomials for \emph{integer} values of $\ell$, i.e. $y(x) = P_\ell(x)$, $\ell \ge 0$, for which $P_\ell(1)=1$, the so-called \emph{Legendre polynomials}.

\section{Legendre's functions and polynomials}

Let us apply Frobenius method to Eq.~\eqref{eq:Legendre} in order to find an analytical solution in the form of a power series that is term-by-term differentiable, even at the poles $\,x = \pm\,1$:\footnote{Note that Eq.~\eqref{eq:Legendre} reduces to $\,\mp\,2 y' +\lambda\,y = 0\,$ at the poles $\,x = \pm\,1$ (which corresponds to $\,\theta = 0, \pi$). This ODE is \emph{separable} and its solution is simply $\,y(x) = e^{\,\pm\,(\lambda/2) \, x}$, which increases without limit when $x \rightarrow \pm \, \infty$, but this is not an issue because the values of $x$ which make sense all belong to the domain $[-1,1]$.}
\begin{equation}
y(x) = \sum_{k=0}^\infty{c_k \, x^k} \, .
\label{eq:serieY}
\end{equation}
From Eq.~\eqref{eq:Legendre}, one finds
\begin{equation}
\left(1 -x^2 \right)\,\sum_{k=2}^\infty{k\,(k-1)\,c_k \, x^{k-2}} - 2 x \,  \sum_{k=1}^\infty{k\,c_k \, x^{k-1}} +\lambda \, \sum_{k=0}^\infty{c_k \, x^k} = 0 \, ,
\end{equation}
which promptly simplifies to
\begin{equation}
\sum_{k=0}^\infty{[\lambda -k\,(k+1)]\,c_k \, x^k} +\sum_{k=2}^\infty{k\,(k-1)\,c_k \, x^{k-2}} = 0 \, .
\end{equation}
The substitution $\tilde{k} = k-2$ in the last sum leads to
\begin{equation}
\sum_{k=0}^\infty{\left\{ [\lambda -k\,(k+1)]\,c_k +(k+1)\, (k+2)\,c_{k+2} \right\} \, x^k} = 0 \, .
\end{equation}
Since this series must be null for all real values of $\,x \in [-1,1]$, then all its coefficients must be null, which leads to the recurrence relation
\begin{equation}
c_{k+2} = -\,\frac{\lambda -k\,(k+1)}{\,(k+1)\, (k+2)} \; c_k \: .
\label{eq:recorrencia}
\end{equation}
In this kind of recurrence, all coefficients of even order are determined by $c_0$ and those of odd order are determined by $c_1$. In particular, if $\,c_0=0\,$ and $\,c_1 \ne 0\,$ then the series in Eq.~\eqref{eq:serieY} will have only odd powers of $x$. Conversely, if $\,c_0 \ne 0\,$ and $\,c_1 = 0\,$ then the series will have only even powers of $x$. The series $\,y(x) = \sum_{k=0}^\infty{c_k \: x^k}$, with the coefficients given by Eq.~\eqref{eq:recorrencia}, are called \emph{Legendre's functions} and the basic condition for it to be a solution of Eq.~\eqref{eq:Legendre} is that the series must \emph{converge} for all $x \in [-1,1]$. However, the recurrence in Eq.~\eqref{eq:recorrencia} has the property that $\,c_{k+2} \approx c_k\,$ as $\,k \rightarrow \infty$, so, heuristically, the infinite series would behave like the geometric series $\,c_\infty \, \sum_{k=0}^\infty{x^k}$, which converges if and only if $\,|x|<1$.\footnote{Within the convergence interval $(-1,1)$, $\sum_{k=0}^\infty{x^k}\,$ converges to $\,1/(1 -x)$.} This creates a dilemma for the endpoints $\,x = \pm 1$ and the only manner to guarantee convergence for all $\,x \in [-1,1]$, including the endpoints, is to require that the series terminates at some integer $\,n \ge 0$. This is done by taking $\,c_n \ne 0\,$ and $\,c_{n+2} = 0$, which makes null all $\,c_{n +2m}$, $m = 1, 2, 3, \ldots$ The choice $\,c_{n+2} = 0\,$ demands the numerator of our recurrence to be null, which occurs \textbf{if and only if} $\,\lambda = n\,(n+1)$.\footnote{It is possible to show (but we'll not do it here) that if $\lambda$ is not a non-negative integer, then there is no bounded solution for Legendre's ODE in $(-1,1)$, except the trivial one $\,y(x)=0$.} This reduces our series solution to a \emph{polynomial function} of degree $n$, with the property that this function is even or odd depending on $n$ to be even or odd, respectively. More explicitly,
\begin{equation}
y_n(x) = c_0 +c_2\,x^2 +c_4\,x^4 +\ldots +c_n\,x^n \, , \quad n \: \mathrm{even} \, ,
\label{eq:PolsEven}
\end{equation}
or
\begin{equation}
y_n(x) = c_1\,x +c_3\,x^3 +c_5\,x^5 +\ldots +c_n\,x^n \, , \quad n \: \mathrm{odd} \, .
\label{eq:PolsOdd}
\end{equation}
From the recurrence in Eq.~\eqref{eq:recorrencia}, it is easy to see that $\,2\,(2n-1)\,c_{n-2} = -\,n\,(n-1)\,c_n$, whose repetition yields all non-null coefficients of order smaller than $n$ as a multiple of $c_n$:
\begin{equation}
c_{n-2k} = \frac{(-1)^k}{2^n\,k!\,} \: \frac{n\,(n-1)\,(n-2)\,\ldots\,(n-(2k-1))}{\,(2n-1)\,(2n-3)\,\ldots\,(2n-(2k-1))\,} \; c_n \, , \quad \forall \: k>0 \, .
\label{eq:InvRecorrencia}
\end{equation}
At this point, all textbooks adopt a `practical normalization' of these polynomials, following the convention that $\,y_n(1) = 1\,$ for all $\,n \ge 0$. As shown in Chap.~6 of Ref.~\cite{lt:Maia2000}, this is attained by choosing $\,c_n = \binom{2n}{n}/2^n = (2n-1)!!/n!\,$ and substituting it in Eq.~\eqref{eq:InvRecorrencia}, which results in\footnote{As usual, given a positive integer $n$, $n!!$ denotes its double-factorial $\,n\,(n-2)\,(n-4)\,\ldots\,1$ (it ends on $2$ if $n$ is even).}
\begin{equation}
c_{n-2k} = \frac{(-1)^k}{2^n\,k!\,} \: \frac{(2n-2k)!}{\,(n-2k)!\,(n-k)!\,} \, ,
\label{eq:cnk}
\end{equation}
which remains valid for $\,k=0$. The explicit polynomial solutions of Legendre's ODE with decreasing powers then reads
\begin{equation}
P_n(x) = \frac{1}{\,2^n} \sum_{k=0}^{\lfloor n/2 \rfloor}{ (-1)^k \; \binom{n}{k} \: \binom{2n -2k}{n} \; x^{n-2k}} \: ,
\label{eq:PnExplicito}
\end{equation}
where $\lfloor z \rfloor$ is the floor of $z$, i.e., the greatest integer $m \le z$.\footnote{For $\,z \ge 0$, this is equivalent to the integer part of $z$.}  These polynomial functions $\,P_n(x)$, which are the unique polynomial solutions of Legendre's differential equation [Eq.~\eqref{eq:origLegendre} or Eq.~\eqref{eq:Legendre}] with $\,y(1) = 1$, are just the \emph{Legendre Polynomials}, named after A.-M.~Legendre, who discovered them in 1782-1783~\cite{Legendre1782}. The first few of them are $\,P_0(x)=1$, $P_1(x)=x$, $P_2(x)=\frac12 \, (3 x^2 -1)$, $P_3(x)=\frac12 \, (5 x^3 -3x)$, $P_4(x) = \frac18\,(35 x^4 -30 x^2 +3)$, and $\,P_5(x) = \frac18 \, (63 x^5 -70 x^3 +15 x)$. These polynomials are plotted in Fig.~\ref{fig2}, below. These few examples suggest that a multiplication of $P_n(x)$ by a suitable power of $2$ would make integer all its non-null coefficients, a relevant feature for applications in Number Theory. In fact, this is possible for all $P_n(x)$ and the correct power is simply the number of factors $2$ in the prime factorization of $n!$, as follows from the choice $\,c_n = (2n-1)!!/n!\,$.\footnote{Legendre himself found a formula for this exponent and, more generally, he found that the exponent of the largest power of a prime $p$ that divides $n!$ is just the $p$-adic valuation of $n!$: $\,\nu_p(n!) = \sum_{k=1}^{\infty} \left\lfloor n/p^k \right\rfloor$. Note that this series has only finitely many nonzero terms, as for every $k$ large enough, $p^k > n$ will lead to $\left\lfloor n/p^k \right\rfloor = 0$. This reduces the infinite sum above to $\,\nu_p(n!) = \sum_{k=1}^{L} \left\lfloor n/p^k \right\rfloor$, where $\,L := \lfloor \log_{p} n \rfloor$.}

\begin{figure}[hbt]
\centering
\scalebox{1.0}{\includegraphics{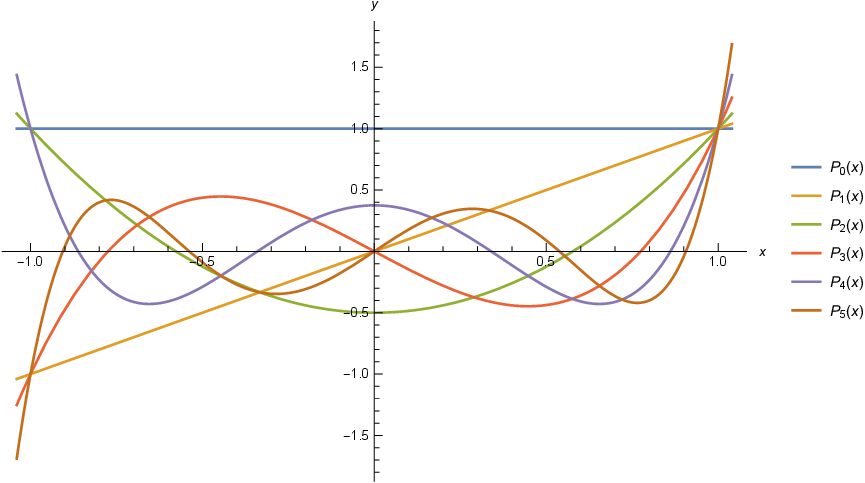}} \label{fig2}
\caption{\small{The first few Legendre polynomials $P_n(x)$, as given by Eq.~\eqref{eq:PnExplicito}. Note that they all obey the `practical normalization' $\,P_n(1)=1$.}}
\end{figure}

Legendre polynomials also have the following useful representation in terms of the $n$-th derivative of $\,(x^2 -1)^n$, first derived by Rodrigues in 1816.\footnote{\label{ft:2F1}$\,P_n(x)$ can also be written in the form ${}_2F_1{[-n,n+1;1;(1-x)/2]}$, as shown in Ref.~\cite[Sec.~5.5]{lt:Capelas}.}

\begin{lema}[Rodrigues' formula for $P_n(x)\,$]  For all non-negative integers $n$, the following representation for Legendre polynomials holds:
\begin{equation}
P_n(x) = \frac{1}{\,2^n \, n!} \; \frac{d^n}{\,dx^n} \left[ (x^2 -1)^n \right] \, .
\end{equation}
\end{lema}

\begin{proof}  For $\,n=0$, since $\,0! = 1\,$ and the $0$-th derivative of a function is defined as the function itself, Rodrigues' formula yields $\,P_0(x) = 1$, which agrees to the practical normalization. For $n>0$, we proceed as follows. Legendre's ODE with explicit integer eigenvalues reads
\begin{equation}
\frac{d}{dx} \left[\left(1-x^2 \right) \, y' \, \right] +n\,(n+1) \, y = 0 \, ,
\label{eq:LegODE}
\end{equation}
Now, note that, for any integer $n>0$,
\begin{equation}
\left(x^2 -1\right) \, \frac{d}{dx} \left(x^2 -1\right)^n = 2 n x \, \left(x^2 -1\right)^n \, .
\end{equation}
By passing the product in the right-hand side to the left and calculating the $(n+1)$-th derivative, one finds
\begin{eqnarray}
\frac{d^{n+1}}{dx^{n+1}} \, \left[ (x^2-1) \, \frac{d}{dx} (x^2-1)^n -2 n x (x^2-1)^n\right] \nonumber \\
= n\,(n+1)\,\frac{d^n}{dx^n}(x^2-1)^n +2\,(n+1)\,x\,\frac{d^{n+1}}{dx^{n+1}} (x^2-1)^n +(x^2-1) \, \frac{d^{n+2}}{dx^{n+2}} (x^2-1)^n \nonumber \\
-2 n\,(n+1) \, \frac{d^n}{dx^n}(x^2-1)^n -2 n\,x\,\frac{d^{n+1}}{dx^{n+1}} (x^2-1)^n \nonumber \\
= -\,n\,(n+1) \, \frac{d^n}{dx^n} (x^2-1)^n +2 x\,\frac{d^{n+1}}{dx^{n+1}} (x^2-1)^n +(x^2-1) \, \frac{d^{n+2}}{dx^{n+2}} (x^2-1)^n \nonumber \\
= -\,\left\{\frac{d}{dx} \left[(1-x^2) \; \frac{d}{dx} \left(\frac{d^n}{dx^n} (x^2-1)^n \right) \right] +n\,(n+1) \: \frac{d^n}{dx^n} (x^2-1)^n \right\} = 0 \, .
\end{eqnarray}
This shows that the function $\,\frac{d^n}{dx^n} (x^2 -1)^n\,$ satisfies Legendre's differential equation, Eq.~\eqref{eq:LegODE}. The practical normalization is obtained from the coefficient of the leading term (of degree $n$), i.e. $\frac{d^n}{dx^n} \left(x^{2n}\right) = 2 n \, (2n-1) \, (2n-2) \, \ldots \, (n+1) \, x^n = \frac{(2n)!}{n!} \, x^n$, which shows that one has to multiply that function by $\,1/\left(2^n \, n!\right)$.
\end{proof}

It promptly follows from Rodrigues' formula that $\,P_n(-\,x) = (-1)^n \, P_n(x)\,$ for all integers $\,n \ge 0$. In particular, $\,P_n(-1) = (-1)^n \, P_n(1) = (-1)^n$, since $\,P_n(1) = 1$. Being both $P_n(1)$ and $P_n(-1)$ non-null, it can be shown (see the Appendix) that the $n$ roots of $\,P_n(x) = 0\,$ are distinct real numbers belonging to the open interval $(-1,1)$. These roots are a crucial part of the \emph{Gauss-Legendre quadrature}, whose modern formulation using orthogonal polynomials was developed by Jacobi in 1826.\footnote{For a $n$-points quadrature, it reads $\,\int_{-1}^{+1}{\,f(x) ~ dx} \approx \sum_{i=1}^n{w_i \: f(\xi_i)}$, where $\,w_i = 2\,/\left[\left(1 -\xi_i^2 \right) \, \left(P'_n(\xi_i)\right)^2 \right]\,$ are the weights, $P_n(x)$ being the Legendre polynomial taken within the practical normalization $\,P_n(1) = 1$, and $\,\xi_i\,$ is the $i$-th Gaussian node, which is just the $i$-th root of $P_n(x)$~\cite[p.~887]{Abramow}. Due to the parity of all $P_n(x)$, $x=0$ is a root of $P_n(x)$ \emph{if and only if} $n$ is odd. The parity of $P_n(x)$ makes the roots to be symmetric with respect to $x=0$, so we need to compute only those roots in the interval $(0,1)$. Incidentally, each one of the $n+1$ subintervals of $(-1,1)$ formed by the $n$ roots of $P_{n-1}(x)$ contain exactly one root of $P_{n}(x)$, an important feature for numerical computations. For a simple proof of this feature, see the Appendix. This quadrature rule is very accurate when $f(x)$ is well-approximated by polynomials on $[-1,1]$, being \emph{exact} for polynomials of degree $2n-1$ or less.}  At the origin $x=0$, the only surviving coefficient in Eqs.~\eqref{eq:PolsEven} and \eqref{eq:PolsOdd} is $c_0$, so we have $\,P_{2 m +1}(0) = 0\,$ and, from Eq.~\eqref{eq:PnExplicito}, $P_{2 m}(0) = c_0 = (-1)^m \, \binom{2m}{m}/2^{2m} = (-1)^m \, (2 m -1)!!/(2 m)!!$.

Additionally, Rodrigues' formula can be taken into account to prove other representations for $P_n(x)$, for instance
\begin{equation}
P_n(x) = \sum_{k=0}^{n}{ \binom{n}{k} \: \binom{n+k}{k} \; \left( \frac{\,x-1}{2}\right)^k} \, ,
\label{eq:Pn2explicito}
\end{equation}
which will be important in the next section, in the derivation of an explicit representation for shifted Legendre polynomials.

Interestingly, Legendre polynomials form a set of \emph{orthogonal} continuous functions over $\,[-1,1]$. We give below a formal proof of this property.

\begin{lema}[Orthogonality of Legendre polynomials]
\label{lema:ortogonal}  Legendre polynomials $P_n(x)$ of distinct degrees are orthogonal over $[-1,1]$, with weighting function $\,w(x) = 1$.
\end{lema}

\begin{proof}
It is well-established in linear algebra that the orthogonality of two continuous functions $f_m(x)$ and $f_n(x)$ over a domain $[a,b]$ is equivalent to the nullity of the inner product $\,\langle f_m | f_n\rangle := \int_a^b{ f_m(x) \: f_n(x) \: w(x) ~ dx}$. On taking $\,w(x)=1$, all we have to do is to show that
\begin{equation}
\int_{-1}^{+1}{ P_m(x) \cdot P_n(x) ~ dx} = 0 \, , \quad \forall \; m \ne n \, .
\end{equation}
Since all $P_n(x)$ satisfy Legendre's ODE, we have
\begin{equation}
\left[\left(1 -x^2 \right) \, P'_n(x) \, \right]' +n\,(n+1) \; P_n = 0
\label{eq:Pn1}
\end{equation}
and
\begin{equation}
\left[\left(1 -x^2 \right) \, P'_m(x) \, \right]' +m\,(m+1) \; P_m = 0 \, .
\label{eq:Pm1}
\end{equation}
Now, multiply both sides of Eq.~\eqref{eq:Pn1} by $P_m(x)$ and those of Eq.~\eqref{eq:Pm1} by $P_n(x)$. A member-to-member subtraction of the resulting expressions will result in
\begin{equation}
\left[\left(1 -x^2 \right) \, \left(P'_m\,P_n -P_m\,P'_n\right) \right]' +(m-n)\,(m+n+1) \; P_m\,P_n = 0 \, .
\label{eq:Pnm}
\end{equation}
Finally, the integration of both sides over the domain $[-1,1]$ will yield
\begin{equation}
\int_{-1}^{+1}{\frac{d}{dx} \left[\left(1 -x^2 \right) \, \left(P'_m\,P_n -P_m\,P'_n\right) \right] ~ dx} +(m-n)\,(m+n+1) \, \int_{-1}^{+1}{P_m\,P_n ~ dx} = 0 \, .
\label{eq:intPnm}
\end{equation}
From the fundamental theorem of calculus, the first integral is always null, so
\begin{equation}
(m-n)\,(m+n+1) \, \int_{-1}^{+1}{P_m\,P_n ~ dx} = 0 \, .
\end{equation}
Clearly, for all non-negative values of $\,m \ne n\,$ the last integral has to be null.
\end{proof}

The orthogonality of $\,P_n(x)\,$ is also the basic feature that allows its derivation by applying the Gram-Schmidt procedure to orthogonalize the powers $\,\left\{1, x, x^2, \ldots \right\}$, as usually done in Linear Algebra textbooks (see, e.g., Sec.~8.3 of Ref.~\cite{Strang5ed}).  In fact, Legendre polynomials can be renormalized in order to form a \emph{complete set of orthonormal polynomials} $\,\hat{P}_n(x)\,$ on the interval $[-1,1]$, i.e.~an \emph{orthonormal basis}.

\begin{lema}[Orthonormality of Legendre polynomials]
\label{lema:ortonormal}  The renormalized Legendre polynomials $\: \hat{P}_n(x) := \sqrt{n +\frac12\,} \; \, P_n(x)\:$ are orthonormal on the interval $\,[-1,1]$.
\end{lema}

\begin{proof}  According to Lemma~\ref{lema:ortogonal}, the polynomials $P_n(x)$ are \emph{orthogonal} over the interval $[-1,1]$. Then, all that rests is to prove that
\begin{equation}
\langle P_n | P_n\rangle = \int_{-1}^{+1}{ P_n^{\,2}(x) ~ dx} = \frac{2}{\,2 n +1} \, , \quad \forall \: n \ge 0 \, .
\label{eq:PnPn}
\end{equation}
For this, we begin noting that $P_n^{\,2}(x)$ is a polynomial of degree $2n$ whose terms all have even exponents, so that it defines an \emph{even function}. Therefore,
\begin{eqnarray}
\int_{-1}^{+1}{ P_n^{\,2}(x) ~ dx} = 2 \, \int_0^1{ P_n(x) \, P_{n}(x) ~ dx} \nonumber \\
= 2\,(-1)^n \, \frac{a_n}{2^n} \, \int_0^1{(x^2 -1)^n ~ dx} \nonumber \\
= 2\,(-1)^n \, \frac{a_n}{2^n} \cdot I_n \: ,
\label{eq:aux1}
\end{eqnarray}
where  $\,a_n = (2n)!/(2^n \, (n!)^2)\,$ and the integral $\,I_n := \int_0^1{(x^2 -1)^n ~ dx}\,$ is such that
\begin{eqnarray}
I_n = \int_0^1{x\,(x^2-1)^n ~dx} -2 n \, \int_0^1{x^2\,(x^2-1)^{n-1} ~ dx} \nonumber \\
= -\,2 n \, \int_0^1{\left[(x^2 -1)^n +(x^2 -1)^{n-1} \right] ~ dx} \nonumber \\
= -\,2 n \, I_n -2 n \, I_{n-1} \, .
\end{eqnarray}
We then have the recurrence relation
\begin{equation}
I_n = -\,\frac{2 n}{\,2 n +1\,} \: I_{n-1} \, .
\end{equation}
Using this and $\,I_0 = \int_0^1{1 ~ dx} = 1$, one easily finds that
\begin{eqnarray}
I_n = (-1)^{n} \, \frac{2 \cdot 4 \cdot 6 \cdots (2n)}{\,3 \cdot 5 \cdot 7 \cdots (2n +1)\,} \nonumber \\
= (-1)^{n} \, \frac{\,[2 \cdot 4 \cdot 6 \cdots(2n)]^{\,2}}{(2n +1)!} \nonumber \\
= (-1)^{n} \: \frac{\,2^{2 n}\,(n!)^2}{(2 n +1)!} \: .
\end{eqnarray}
On substituting this in Eq.~\eqref{eq:aux1}, one finally finds that
\begin{eqnarray}
\int_{-1}^1{P_n^{\,2}(x) ~ dx} = \frac{(-1)^{n}}{2^{n-1}} \cdot \frac{(2n)!}{2^{n}\,(n!)^2} \cdot (-1)^{n} \, \frac{2^{2n} \, (n!)^2}{(2n+1)} \nonumber \\
= \frac{2}{\,2n +1\,} \, .
\end{eqnarray}
\end{proof}

Our Lemmas~\ref{lema:ortogonal} and \ref{lema:ortonormal} can now be expressed in a single \emph{orthonormality relation}.
\begin{teo}[Orthonormality relation for $P_n(x)\,$]
\label{teo:ortoPn}  For any non-negative integers $m$ and $n$,
\begin{equation}
\langle P_m | P_n \rangle = \int_{-1}^{+1}{ P_m(x) \cdot P_n(x) ~ dx} = \frac{2}{\,2 n +1} \; \delta_{m n} \: ,
\label{eq:basePn}
\end{equation}
where $\,\delta_{mn}\,$ is the Kronecker's delta ($1$ if $m=n$, $0$ otherwise).
\end{teo}
Then, the \emph{orthonormal} Legendre polynomials are simply\footnote{Of course, $\int_{-1}^{1}{\hat{P}_m(x) \cdot \hat{P}_n(x) ~ dx} = \delta_{m n}$.}
\begin{equation}
\hat{P}_n(x) = \sqrt{\frac{2n +1}{2}\,} ~ P_n(x) \, .
\end{equation}
Since the renormalized Legendre polynomials $\hat{P}_n(x)$ form a basis for the vector space $\,C_{[-1,1]}^{\,0}\,$ of continuous functions over $[-1,1]$, then the Fourier-Legendre series expansion of any function $f(x)$ continuous by parts on this interval will converge to $\,[f(x^{-}) +f(x^{+})]/2$, as long as $f(x)$ has lateral derivatives at $x$, according to Fourier's convergence theorem. Also, for all $\,x \in [-1,1]\,$ where $f(x)$ is continuous, Fourier's theorem allows us to conclude that
\begin{equation}
f(x) = \sum_{n=0}^\infty{\,a_n \, P_n(x)} \, ,
\end{equation}
where the Fourier coefficients are, according to Eq.~\eqref{eq:PnPn},
\begin{equation}
a_n = \frac{\langle P_n | f \rangle}{\langle P_n | P_n \rangle} = \frac{\,\int_{-1}^1{P_n(x) \: f(x) ~ dx}\,}{2/(2n+1)} = \left(n +\frac12 \right) \, \int_{-1}^{\,1}{P_n(x) \: f(x) ~ dx} \, .
\end{equation}
In the computation of these coefficients, it is useful to know that $\,\int_{-1}^1{P_0(x) \: dx} = \int_{-1}^1{1 \: dx} = 2$, $\int_{-1}^1{P_n(x) \: dx} = 0\,$ for all $n \ge 1$, $\int_{-1}^1{x^m \, P_n(x) \: dx} = 0\,$ for all positive integers $\,m < n$, and $\,\int_{-1}^1{x^n \: P_n(x) ~ dx} = 2^{n+1} \: (n!)^2/(2n+1)!$.\footnote{It is easy to show, starting from the Rodrigues formula and integrating by parts, that $\,\int_{-1}^1{x^m \, P_n(x) \: dx} = 0\,$ for all positive integers $\,m < n$. More generally, $\int_{-1}^1{p(x) \: P_n(x) ~ dx} = 0\,$ for any polynomial $p(x)$ with a degree smaller than $n$, as follows from the fact that the Legendre polynomials $\,P_0(x), P_1(x), \ldots, P_n(x)\,$ are orthogonal and linearly independent over $[-1,1]$, composing a vector space of dimension $n+1$.}

We conclude this section mentioning that all Legendre polynomials and their properties can also be derived from the \emph{generating function}
\begin{equation}
\frac{1}{\,\sqrt{1 -2 x\,t +t^2\,}} = \sum_{n=0}^\infty{P_n(x) \: t^n} \: ,
\label{eq:geratriz}
\end{equation}
including the Bonnet's recurrence
\begin{equation}
P_{n+1}(x) = \frac{2n+1}{n+1} \: P_n(x) -\frac{n}{n+1} \: P_{n-1}(x) \, ,
\end{equation}
as done in some textbooks (see, e.g., Sec.~18.1 of Ref.~\cite{Riley3ed}). Indeed, the form of the above generating function comes from the \emph{multipole expansion} of the potential in electrostatics and gravitational problems, which is how the polynomials $\,P_n(\cos{\theta})\,$ were discovered by Legendre in his original work~\cite{Legendre1782}. In fact, there in that paper Legendre realized that the Newtonian potential can be expanded as
\begin{eqnarray}
\frac{1}{\,\left| \vec{r} -\vec{r}\,' \right|\,} = \frac{1}{\,\sqrt{r^2 +{r'}^2 -2 r r' \, \cos{\theta}}\,} \nonumber \\
= \frac{1}{\,r\,\sqrt{1 -2 \, (r'/r) \, \cos{\theta} +(r'/r)^2 }\,} \nonumber \\
= \frac{1}{r} \, \sum_{n=0}^\infty{\left(\frac{r'}{r}\right)^n \, P_n(\cos{\theta})} \, ,
\end{eqnarray}
where $r$ and $r'$ are the lengths of the vectors $\vec{r}$ and $\vec{r}\:'$, respectively, and $\theta$ is the smaller angle between them.\footnote{Here, we are assuming that $\,r' < r$, in order to guarantee convergence. } The substitutions $\,r'/r = t\,$ and $\,\cos{\theta} = \vec{r} \cdot \vec{r}\:'/(r\,r') = x\,$ lead to Eq.~\eqref{eq:geratriz}.

\section{Shifted Legendre Polynomials}

The shifted Legendre Polynomial is defined as $\,\tilde{P}_n(x) := P_n(1 -2 x)$, being normalized on the interval $[0,1]$.\footnote{The reader will find in some textbooks the alternative definition $\,\tilde{P}_n(x) := P_n(2 x -1)$, but we shall follow here that adopted by Beukers in Refs.~\cite{Beukers79,Beukers1980}.}  Of course, we can derive all properties of this new class of polynomials from those already known for $P_n(x)$ by substituting $\,x\,$ by $\,1 -2x$. For instance, the `practical normalization' yields $\,\tilde{P}_n(0) = P_n(1 -2 \cdot 0) = P_n(1) = 1$. For the other endpoint, one has $\,\tilde{P}_n(1) = P_n(1 -2 \cdot 1) = P_n(-1) = (-1)^n$. In fact, this can be generalized in the form of a \emph{reflection formula} around $x=1/2$ valid for all $\,x \in [0,1]$, namely
\begin{eqnarray}
\tilde{P}_n(1-x) = \tilde{P}_n(u) = P_n(1-2u) \nonumber \\
= P_n(1-2\,(1-x)) = P_n(-1+2x) \nonumber \\
= P_n(-(1-2x)) = (-1)^n \, P_n(1-2x) \nonumber \\
= (-1)^n \, \tilde{P}_n(x) \, , \quad \forall \: n \ge 0 \, .
\end{eqnarray}
For $x=1/2$, $\tilde{P}_n(1/2) = P_n(1 -2 \cdot 1/2) = P_n(0)$, which evaluates to $\,0\,$ if $n$ is odd and $\,(-1)^{n/2} \, (n-1)!!/n!!\,$ if $n$ is even.

From Eq.~\eqref{eq:PnExplicito}, the explicit expression of $\,\tilde{P}_n(x)\,$ is given by
\begin{eqnarray}
\tilde{P}_n(x) = P_n(1 -2 x) = \frac{1}{\,2^n} \sum_{k=0}^{\lfloor n/2 \rfloor}{ \frac{(-1)^k}{k!} \; \frac{(2n-2k)!}{\,(n-2k)! \: (n-k)!\,} \; (1 -2 x)^{n-2k}} \nonumber \\
= \frac{1}{\,2^n} \sum_{k=0}^{\lfloor n/2 \rfloor}{ \frac{(-1)^k}{k!} \; \frac{(2n-2k)!}{\,(n-2k)! \: (n-k)!\,} \; \sum_{j=0}^{n-2k}{\binom{n-2k}{j}\,1^j \, (-2 x)^{n-2k-j}}} \nonumber \\
= \sum_{k=0}^{\lfloor n/2 \rfloor}{ \frac{(-1)^k}{2^{2k} \, k!} \; \frac{(2n-2k)!}{(n-k)!}\; \sum_{j=0}^{n-2k}{\frac{1}{j!\,(n-2k-j)!} \, \frac{1}{2^j} \, (-x)^{n-2k-j}}} \nonumber \\
= \frac{1}{2^n} \, \sum_{k=0}^{\lfloor n/2 \rfloor}{ \frac{(-1)^k}{k!} \; \frac{(2n-2k)!}{(n-k)!} \: \sum_{m=0}^{n-2k}{\frac{2^m}{m!\,(n-2k-m)!} \, (-x)^m}} \, ,
\label{eq:tildePnExplicito}
\end{eqnarray}
where we have substituted $n-2k-j$ by $m$ in the last step.\footnote{For $\,n=0$, these sums reduce to only one term, promptly revealing that $\,\tilde{P}_0(x) = 1$, as expected.}  Since this sum is somewhat complex from the computational point of view, let us take into account an explicit representation of $\tilde{P}_n(x)$ used by Beukers in Ref.~\cite{Beukers1980} to show that \emph{all coefficients} of $\,\tilde{P}_n(x)\,$ are \emph{non-null integers} with alternating signs.

\begin{teo}[Simpler explicit representation of $\tilde{P}_n(x)\,$]  For any integer $n \ge 0$, one has
\label{teo:simpler}
\begin{equation}
\tilde{P}_n(x) = \sum_{m=0}^n{\binom{n}{m} \, \binom{n+m}{m} \, (-\,x)^m} \, .
\label{eq:BkTildePn}
\end{equation}
\end{teo}

\begin{proof}
It is enough to substitute $\,X = 1 -2 x\,$ in Eq.~\eqref{eq:Pn2explicito}, which leads to
\begin{eqnarray}
\tilde{P}_n(x) = P_n(X) = \sum_{k=0}^n{\binom{n}{k} \: \binom{n+k}{k} \; \left( \frac{\,X-1}{2}\right)^k} \nonumber \\
= \sum_{k=0}^{n}{ \binom{n}{k} \: \binom{n+k}{k} \; (-\,x)^k} \, .
\end{eqnarray}
\end{proof}

From Theorem~\ref{teo:simpler}, the first few of these polynomials are: $\tilde{P}_0(x) = 1$, $\tilde{P}_1(x) = 1 -2\,x$, $\tilde{P}_2(x) = 1 -6\,x +6\,x^2$, $\tilde{P}_3(x) = 1 -12\,x +30\,x^2 -20\,x^3$, $\tilde{P}_4(x) = 1 -20\,x +90\,x^2 -140\,x^3 +70\,x^4$, and $\,\tilde{P}_5(x) = 1 -30\,x +210\,x^2 -560\,x^3 +630\,x^4 -252\,x^5$. They are plotted in Fig.~\ref{fig3}, below.

\begin{figure}[hbt]
\centering
\scalebox{1.0}{\includegraphics{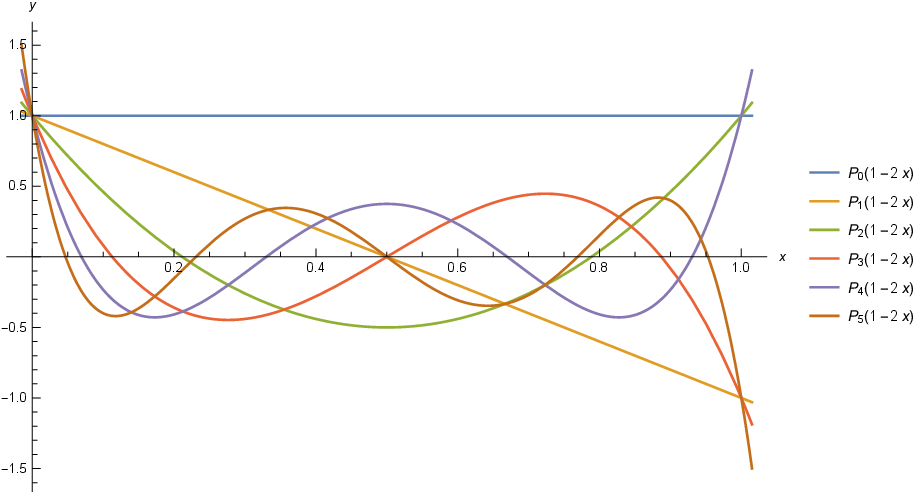}} \label{fig3}
\caption{\small{The first few shifted Legendre polynomials $\,\tilde{P}_n(x)$, as found from Eq.~\eqref{eq:BkTildePn}, which obey the `practical normalization' $\,\tilde{P}_n(0)=1$.}}
\end{figure}

The Polynomial $\tilde{P}_n(x)$ is the main tool of Beukers-like irrationality proofs, which is due to the particular form of its \emph{Rodrigues formula}.

\begin{teo}[Rodrigues' formula for $\tilde{P}_n(x)\,$]
\label{lem:RodriguesTildePn}  For all integers $\,n \ge 0$,
\begin{equation}
\tilde{P}_n(x) = \frac{1}{n!} ~ \frac{d^n}{dx^n} \left[\,x^n \, (1-x)^{n\,} \right] \, .
\end{equation}
\end{teo}

\begin{proof} For $\,n=0$, Rodrigues formula returns $\,\tilde{P}_0(x) = (1/0!) \times \left[x^0 \, (1-x)^0 \right] = 1$, in agreement to our Eq.~\eqref{eq:BkTildePn}. For $\,n>0$, we have already shown that $\,P_n(X) = \frac{1}{2^n \, n!}\frac{d^n}{dX^n} \, \left[\left(X^2-1\right)^n \right]$ for the usual Legendre polynomial. On taking $X=1-2x$, one finds
\begin{equation}
\tilde{P}_n(x) := P_n(1-2x) = \frac{1}{2^n \, n!} \, \frac{d^n}{d(1-2x)^n} \, \left[ (1 -2 x)^2 -1 \right]^n \, .
\end{equation}
Since $\,d(1-2x) = -2 \, dx$, then $\,d(1-2x)^n = (-2)^n \, dx^n$  and, being $\left[(1-2x)^2 -1\right]^n = \left(4 x^2 -4 x\right)^n = 4^n \, \left(x^2-x\right)^n$, it follows that
\begin{eqnarray}
\tilde{P}_n(x) = \frac{1}{2^n \, n!} \, \frac{d^n}{(-2)^n dx^n} \, \left[4^n \, \left(x^2-x\right)^n \right] \nonumber \\
= \frac{4^n}{2^n \, 2^n \, n!} \, \frac{d^n}{dx^n} \, \left[\left(x -x^2\right)^n\right] \nonumber \\
= \frac{1}{n!} ~ \frac{d^n}{dx^n} \, \left[ \left(x -x^2 \right)^n\right] \nonumber \\
= \frac{1}{n!} ~ \frac{d^n}{dx^n} \, \left[\,x^n \, (1-x)^{n\,} \right] \, .
\end{eqnarray}
\end{proof}

We can apply the Leibnitz rule for the derivative of the product of two functions in order to derive another representation for $\tilde{P}_n(x)$.

\begin{teo}[Another representation for $\tilde{P}_n(x)\,$]  For all integers $\,n \ge 0$,
\begin{equation}
\tilde{P}_n(x) = \sum_{k=0}^n{\binom{n}{k}^2 \; \sum_{j=0}^k{\binom{k}{j} \, (-x)^{n-k+j} } } \: .
\end{equation}
\end{teo}

\begin{proof}  Leibnitz rule for the derivative of the product of two functions reads $\,(u\,v)' = u'\,v +u\,v'$. The repeated application of this rule leads to
\begin{equation}
(u\,v)^{(k)} = \sum_{j=0}^k{\binom{k}{j} \: u^{(j)} \; v^{(k-j)}} \, .
\label{eq:LeibRule}
\end{equation}
On applying this general rule to the derivative in Rodrigues' formula, one finds
\begin{eqnarray}
\tilde{P}_n(x) = \frac{1}{n!} \, \left[x^n \, (1-x)^n\right]^{(n)} = \frac{1}{n!} ~ \sum_{j=0}^n{\binom{n}{j} \, \left(x^n \right)^{(j)} \, \left[ (1 -x)^n \right]^{(n-j)}} \nonumber \\
= \frac{1}{n!} ~ \sum_{k=0}^n{\binom{n}{k} \, \frac{n!}{(n-k)!} \, x^{n-k} \, \frac{n!}{k!} \, (1-x)^k \, (-1)^{n-k}} \nonumber \\
= \sum_{k=0}^n{\binom{n}{k} \, \binom{n}{k} \, (-x)^{n-k} \, (1-x)^k} \, .
\end{eqnarray}
From the binomial theorem, we known that $\,(1-x)^k = \sum_{j=0}^k{\binom{k}{j} \, 1^{k-j}\,(-x)^j} = \sum_{j=0}^k{\binom{k}{j} \, (-x)^j}$. The substitution of this result in the last sum, above, completes the proof.
\end{proof}

This representation confirms that all coefficients of $\tilde{P}_n(x)$ are non-null integers.\footnote{The shifted Legendre polynomials have a nice representation in the hypergeometric form, namely $\,\tilde{P}_n(x) = {}_2F_1{[-n,n+1;1;x]}$, found by substituting $x$ by $1-2x$ in that hypergeometric form mentioned in the previous section (see Footnote~\ref{ft:2F1}). We shall not explore this form here.}

The polynomials $\tilde{P}_n(x)$ are \emph{orthogonal} over $[0,1]$, with weighting function $\,w(x)=1$, and, over this interval, they can be renormalized in order to compose an \emph{orthonormal basis}. In order to show this, it is enough to determine their \emph{orthonormality} relation, as follows:
\begin{eqnarray}
\langle \tilde{P}_m | \tilde{P}_n \rangle = \int_0^1{ \tilde{P}_m(x) \cdot \tilde{P}_n(x) ~ dx} = \int_1^{-1}{ P_m(u) \cdot P_n(u) ~ \frac{du}{(-2)}} \nonumber \\
= \frac{1}{2} \, \int_{-1}^{\,1}{ P_m(u) \cdot P_n(u) ~ du} \nonumber \\
= \frac{1}{\,2n+1} \; \delta_{m n} \: ,
\label{eq:PnTilOrtonormal}
\end{eqnarray}
where we have made use of the simple substitution $\,u = 1 -2x$. The renormalized polynomials are then given by $\,\hat{\tilde{P}}_n(x) = \sqrt{2n +1} \; \tilde{P}_n(x)$. These \emph{orthonormal} polynomials form a basis for the vector space $\,C_{[0,1]}^{\,0}\,$ of continuous functions over $[0,1]$, so the corresponding Fourier-Legendre series for any function $f(x)$ continuous by parts on this interval will converge to $\,[f(x^{-}) +f(x^{+})]/2$, as long as $f(x)$ has lateral derivatives at $x$, according to Fourier's convergence theorem. Also, for all $\,x \in [0,1]\,$ where $f(x)$ is continuous, Fourier's theorem guarantees that
\begin{equation}
f(x) = \sum_{n=0}^\infty{\,b_n \, \tilde{P}_n(x)} \, ,
\end{equation}
where the Fourier coefficients are, according to Eq.~\eqref{eq:PnTilOrtonormal},
\begin{equation}
b_n = \frac{\langle \tilde{P}_n | f \rangle}{\langle \tilde{P}_n | \tilde{P}_n \rangle} = \frac{\,\int_0^1{\tilde{P}_n(x) \: f(x) ~ dx}\,}{1/(2n+1)} = (2 n +1) \, \int_0^{\,1}{\tilde{P}_n(x) \: f(x) ~ dx} \, .
\end{equation}
In the computation of these coefficients, it is useful to know that $\,\int_0^1{\tilde{P}_0(x) \: dx} = \int_0^1{1 \: dx} = 1$, $\int_0^1{\tilde{P}_n(x) \: dx} = 0\,$ for all $n \ge 1$, $\int_0^1{x^m \, \tilde{P}_n(x) \: dx} = 0\,$ for all positive integers $\,m < n$, and
\begin{equation}
\int_0^1{x^n \: \tilde{P}_n(x) ~ dx} = (-1)^n \: \frac{(n!)^2}{\,(2n+1)!} \, .
\end{equation}
For a proof of this last integral, it is enough to substitute $\,u = 1 -2x\,$ in the corresponding integral for $P_n(x)$, namely
\begin{equation}
\int_{-1}^1{u^n \: P_n(u) ~ du} = 2^{n+1} \: \frac{(n!)^2}{\,(2n+1)!} \, ,
\end{equation}
which results in
\begin{equation}
\int_0^1{(1 -2x)^n \: \tilde{P}_n(x) ~ dx} = 2^n \: \frac{(n!)^2}{\,(2n+1)!} \, ,
\end{equation}
and then to apply the binomial theorem to expand the factor $(1-2x)^n$. This leads to
\begin{equation}
\sum_{k=0}^n{\binom{n}{k} \, (-1)^k \, 2^k \, \int_0^1{x^k \: \tilde{P}_n(x) ~ dx}} = 2^n \: \frac{(n!)^2}{\,(2n+1)!} \, .
\end{equation}
Since all terms of this sum are null, except the last (i.e., that for $k=n$), it reduces to
\begin{equation}
(-1)^n \; 2^n \, \int_0^1{x^n \: \tilde{P}_n(x) ~ dx} = 2^n \: \frac{(n!)^2}{\,(2n+1)!} \, ,
\end{equation}
which promptly simplifies to the desired result.

An important application of shifted Legendre polynomials is in irrationality proofs using Beukers-like integrals over $[0,1]$. The choice of $\tilde{P}_n(x)$ to develop those proofs comes from the possibility of performing integration by parts easily.

\begin{teo}[Integration by parts with $\,\tilde{P}_n(x)\,$]
\label{teo:intPnx}
\; Given an integer $\,n \ge 0\,$ and a function $\,f\!: [0,1] \rightarrow \mathds{R}\,$ of class $\,\mathcal{C}^n$, one has
\begin{equation}
\int_0^1{\tilde{P}_n(x) \, f(x) \: dx} = \frac{(-1)^n}{n!} \: \int_0^1{x^n \, (1-x)^n ~ \frac{d^n f}{dx^n} ~ dx} \, .
\label{eq:intIn}
\end{equation}
\end{teo}

\begin{proof} \, For $\,n=0$, since $\tilde{P}_0(x)=1\,$ one finds $\,\int_0^1{\tilde{P}_0(x) \: f(x) ~ dx} = \int_0^1{f(x) ~ dx}$, which agrees with Eq.~\eqref{eq:intIn}, because $\,f^{(0)}(x)=f(x)$. For $\,n>0$, the proof is a sequence of integration by parts (but it suffices to make the first one and observe the pattern). From Rodrigues' formula (our Lemma~\ref{lem:RodriguesTildePn}), it follows that
\begin{eqnarray}
I_n := \int_0^1{\tilde{P}_n(x) \, f(x) \: dx} = \int_0^1{\frac{1}{n!} \: \frac{d^n}{dx^n}\left[x^n\,(1-x)^n \right] \: f(x) ~ dx}  \nonumber \\
= \frac{1}{n!} \, \int_0^1{\frac{d}{dx} \left\{ \frac{d^{n-1}}{dx^{n-1}}\left[x^n\,(1-x)^n \right] \right\} \, f(x) ~ dx} \, .
\label{eq:InBoa}
\end{eqnarray}
Since $\,\int{u \, dv} = u \, v -\int{v \, du}$, then let us choose $\,u=f(x)\,$ and $\,dv = \frac{d}{dx}\{\ldots\} \: dx$. With this choice, $du = f'(x) \, dx\,$ and $\,v = \{\ldots\}$, so
\begin{eqnarray}
n! \: I_n = \left[f(x) \, \frac{d^{\,n-1}}{dx^{n-1}}\left(x^n \, (1-x)^n \right)\right]_0^1 -\int_0^1{\frac{d^{n-1}}{dx^{n-1}}\left[x^n\,(1-x)^n\right] \: f'(x)} \: dx  \nonumber \\
= \left[f(x) \, \sum_{k=0}^{n-1}{\binom{n-1}{k}\,(x^n)^{(k)} \, \left[(1-x)^n\right]^{(n-1-k)} } \right]_0^1 -\int_0^1{\left[x^n\,(1-x)^n\right]^{(n-1)} \, f'(x)} \: dx \, , \quad
\end{eqnarray}
where we have made use of the Leibnitz rule for the higher derivatives of the product of two functions, our Eq.~\eqref{eq:LeibRule}. On calculating these derivatives, one finds
\begin{eqnarray}
n! \: I_n = \left[f(x) \, \sum_{k=0}^{n-1}{\binom{n-1}{k} \, \frac{n!}{(n-k)!} \, x^{n-k} \, \frac{n!}{(k+1)!} \, (-1)^k \, (1-x)^{k+1} } \right]_0^1  \nonumber \\
-\int_0^1{\left[x^n\,(1-x)^n\right]^{(n-1)} \, f'(x)} \: dx  \nonumber \\
= \, \left[ f(1) \times 0 - f(0) \times 0 \right] -\int_0^1{\left[x^n\,(1-x)^n\right]^{(n-1)} \, f'(x)} \: dx  \nonumber \\
= -\int_0^1{\left[x^n\,(1-x)^n\right]^{(n-1)} \, f'(x)} \: dx \, , \;
\end{eqnarray}
where the zeros multiplying $\,f(1)\,$ and $\,f(0)\,$ come from the presence of factors $\,(1-x)\,$ and $\,x$, respectively, with positive exponents in every terms of the finite sum on $k$. Therefore, the overall result of the first integration by parts is the transference of a derivative $\,d/dx\,$ from $\,\tilde{P}_n(x)\,$ to the function $f(x)$ and a change of sign. Each further integration by parts will have the same effect.
\end{proof}

Once $f(x)$ is chosen such that $\,\int_0^1{f(x) ~ dx}\,$ is a simple function of the number $\xi$ whose irrationality we are intending to prove (usually a linear form in $\mathds{Q}$), the integrals $I_n$ defined in Eq.~\eqref{eq:InBoa} are such that
\begin{eqnarray}
\left| I_n \right| = \frac{1}{n!} \: \left| \int_0^1{x^n \, (1-x)^n ~ \frac{d^n f}{dx^n} ~ dx} \right| \nonumber \\
= \frac{1}{n!} \: \left| \int_0^1{g^n(x) \: h(x) ~ dx} \right| \nonumber \\
\le \frac{1}{n!} \: M^n \, \left| \int_0^1{h(x) ~ dx} \right| \, ,
\end{eqnarray}
where $M$ is the maximum of $g(x)$ over $[0,1]$. Then, if $M$ is small enough such that the last expression tends to $0$ as $n \rightarrow \infty$, even when it is multiplied by certain integers $Q_n$, then an irrationality proof for $\xi$ can, {\it a priori}, be developed. For more details on this kind of proof, see, e.g., Refs.~\cite{Beukers79,Huyle2001,Lima2021} and references therein.


\vspace{1.2cm}

\appendix

\section*{Appendix}

\vspace{0.5cm}

\section{The roots of $\,P_n(x) = 0\,$ are \emph{distinct} numbers in $\,(-1,1)$}

\setcounter{equation}{0}

\renewcommand{\theequation}{A.\arabic{equation}}

Let us show that, for any integer $n>0$, all the $n$ roots of $\,P_n(x) = 0\,$ are \emph{distinct real numbers} belonging to the open interval $(-1,1)$.

\begin{proof}
For this, we begin noting that the zeros of the real function
\begin{equation}
f(x) = (x^2 -1)^n = (x+1)^n\,(x-1)^n \, ,
\end{equation}
$n$ being any positive integer, are $\,x = \pm 1$, each with multiplicity $n$. From Rolle's theorem, we know that there is at least one $\,c \in (-1,1)\,$ for which $\,f'(x) = 0$. By direct evaluation of the derivative, it is easy to check that such root lies at $\,c=0$, besides the roots at $\,x = \pm 1$, each with multiplicity $n-1$. Analogously, $f''(x) = 0\,$ has two roots, one in $\,(-1,0)\,$ and the other in $\,(0,1)$, besides the roots at $\,x = \pm 1$, each with multiplicity $n-2$. The repetition of this reasoning for the higher derivatives of $f(x)$ will lead to the application of Rolle's theorem to $\,f^{(n-1)}(x)$, which will yield simple zeros at $n-1$ distinct points in $(-1,1)$, besides those zeros at $\,x = \pm 1$, now with multiplicity $1$ (i.e., simple zeros). This implies that the equation
\begin{equation}
\frac{d^n}{\,dx^n} \left[ (x^2 -1)^n \right] = 0
\end{equation}
has $n$ distinct roots, \emph{each lying between two consecutive zeros of} $f^{(n-1)}(x)$. Hence, according to Rodrigues formula,
\begin{equation}
P_n(x) = \frac{1}{\,2^n \, n!} \; \frac{d^n}{\,dx^n} \left[ (x^2 -1)^n \right] \, ,
\end{equation}
which is a multiple of $\,f^{(n)}(x)$, will have $n$ distinct real roots in $(-1,1)$.
\end{proof}

\newpage

\section{Legendre polynomials on any interval $[a,b]$}

From the results derived in the main text, we can create classes of shifted Legendre polynomials normalized over any real interval $[a,b]$, $b>a$.  A simple, and computationally efficient way is to perform a direct transformation of $P_n(x)$, originally defined on $[-1,1]$, an interval over which they are orthogonal, with weighting function $\,w(x)=1$, to a new polynomial $\tilde{P}_n(x)$ normalized over $[a,b]$, thus maintaining the orthogonality property. For this, it is enough to take
\begin{equation}
\tilde{P}_n(x) = P_n\left(\frac{2 x -a -b}{b -a}\right) = P_n(\alpha \, x -\beta) \, ,
\label{eq:defPntilab}
\end{equation}
where $P_n(x)$ is the original Legendre polynomial, normalized over $[-1,1]$, $\alpha = 2/(b-a)$, and $\,\beta = (b+a)/(b-a)$. It is easy to see that these new polynomials $\tilde{P}_n(x)$ will be \emph{orthogonal} over $[a,b]$ and that they will have an \emph{orthonormality relation}.

\begin{teo}[Orthonormality relation for $\tilde{P}_n(x)$ over $(a,b)\,$]  For any real interval $[a,b]$, $b>a$, and any non-negative integers $m$ and $n$, the following orthonormality relation holds:
\begin{equation}
\int_a^b{\tilde{P}_m(x) \; \tilde{P}_n(x) ~ dx} = \frac{\,b -a\,}{\,2n +1} \;\, \delta_{mn} \: .
\end{equation}
\end{teo}

\begin{proof}
From our Theorem~\ref{teo:ortoPn}, the original Legendre polynomials $P_n(x)$ satisfy
\begin{equation}
\int_{-1}^{1}{P_m(x) \: P_n(x) ~ dx} = \frac{2}{2n+1} \; \delta_{mn} \, .
\end{equation}
For any other interval $[a,b]$, $b>a$, it is enough to substitute
\begin{equation}
x = \frac{\,b -a\,}{2} \; t +\frac{\,a +b}{2} \: .
\label{eq:substFim}
\end{equation}
From the orthonormality relation for $P_n(x)$ on $[-1,1]$, above, one finds
\begin{eqnarray}
\frac{2}{2n+1} \: \delta_{mn} = \int_{-1}^{1}{P_m(t) \: P_n(t) ~ dt} \nonumber \\
= \int_a^b{P_m\!\left(\frac{2x-a-b}{b-a} \right) \: P_n\!\left(\frac{2x-a-b}{b-a}\right) \: \frac{2}{\,b-a\,} ~ dx} \, .
\end{eqnarray}
From Eq.~\eqref{eq:defPntilab}, this reduces to
\begin{equation}
\int_a^b{\tilde{P}_m(x) \; \tilde{P}_n(x) ~ dx} = \frac{\,b-a\,}{2} \; \frac{2}{\,2n+1\,} \; \delta_{mn} = \frac{b-a}{\,2n +1\,} \;\, \delta_{mn} \: .
\end{equation}
\end{proof}

The integration by parts of $\tilde{P}_n(x)$ multiplied by a suitable function $f(x)$, as seen in our Theorem~\ref{teo:intPnx}, can be generalized as follows.

\begin{teo}[Integration by parts with $\,d^{\,n} f/dx^n\,$]
\label{teo:intPnxab}
\; Given an integer $\,n \ge 0\,$ and two functions $\,f,g\!: [a,b] \rightarrow \mathds{R}\,$ of class $\,\mathcal{C}^n$, if $\,\left[ g(x) \cdot f^{(k)}(x) \right]_a^b = 0\,$ for all $\,1 \le k < n$, then
\begin{equation}
\int_a^b{\frac{d^{\,n} f}{dx^n} \; g(x) ~ dx} = (-1)^n \, \int_a^b{f(x) \: \frac{d^{\,n} g}{dx^n} ~ dx} \, .
\label{eq:intInab}
\end{equation}
\end{teo}

\begin{proof} \, Analogously to the proof of Theorem~\ref{teo:intPnx}, it suffices to make the first integration by parts. Being $\,du = \frac{dg}{dx} ~ dx\,$ and $\,v = \int{\frac{df}{dx} ~ dx} = f(x)$, one has
\begin{eqnarray}
J_1 := \int_a^b{\frac{df}{dx} \: g(x) ~ dx} = \int_a^b{f'(x) \: g(x) ~ dx} = \left[ f(x) \, g(x) \right]_a^b -\int_a^b{f(x) \: g'(x) ~ dx} \nonumber \\
= 0 -\int_a^b{f(x) \: g'(x) ~ dx} \, .
\end{eqnarray}
This can be viewed as the result for the case $n=1$. Then, for any integer $n>1$ one finds
\begin{eqnarray}
J_n := \int_a^b{\frac{d^{\,n} f}{dx^n} \; g(x) ~ dx} = \int_a^b{f^{(n)}(x) \; g(x) ~ dx} \nonumber \\
= \left[ g(x) \; f^{(n-1)}(x) \right]_a^b -\int_a^b{f^{(n-1)}(x) \; g'(x) ~ dx} \nonumber \\
= 0 -\int_a^b{f^{(n-1)}(x) \;\, \frac{dg}{dx} ~ dx} \, .
\end{eqnarray}
Therefore, the overall result of each integration by parts is the transference of a derivative $\,d/dx\,$ from $\,f(x)\,$ to the function $g(x)$ and a change of sign. Each further integration by parts will have the same effect.
\end{proof}

For those readers more interested in Gauss-Legendre quadrature, any integral over $[a, b]$ can be easily changed into an integral over $[-1,1]$ in view to apply the original quadrature rule. This change of interval is done by taking into account the same simple substitution given in Eq.~\eqref{eq:substFim}. This yields
\begin{eqnarray}
\int_a^b{f(x) ~ dx} = \int_{-1}^{\,1}{\,f\!\left(\frac{\,b -a}{2} \: t +\frac{\,a +b}{2} \right) \: \frac{dx}{dt} ~ dt} \nonumber \\
= \frac{\,b -a\,}{2} \: \int_{-1}^{\,1}{\,f\!\left(\frac{\,b-a}{2} \: t +\frac{\,a+b}{2} \right) \, dt} \, .
\end{eqnarray}
Applying the $n$-point Gaussian quadrature to this new integral then results in the approximation
\begin{equation}
\int_a^b{f(x) ~ dx} \approx \frac{\,b-a}{2} \: \sum_{i=1}^n{w_i \; f\!\left(\frac{b-a}{2} \: \xi_{\,i} +\frac{a+b}{2}\right)} \, .
\end{equation}

\section*{Acknowledgement}

The author thanks his more involved students Jo\~{a}o Bitencourt, Matheus Diniz, and Vin\'{i}cius Moraes for the assistance in the organization of his hand-written lecture notes.

\end{document}